\documentclass[12pt]{article}
\usepackage{amsthm,amsmath,indentfirst,amssymb}
\usepackage{algorithm,algpseudocode}
\usepackage{color}
\newtheorem{lemma}{Lemma}[section]
\newtheorem{theorem}[lemma]{Theorem}
\newtheorem{corollary}[lemma]{Corollary}

\newtheorem{definition}[lemma]{Definition}

\usepackage{amsmath,amsthm,amsfonts,amssymb,indentfirst}
\usepackage{algorithm, algpseudocode}
\usepackage{graphicx}
\usepackage{subcaption}
\usepackage{color}

\begin{document}

\title{Approximation Algorithm for Minimum Weight
$(k,m)$-CDS Problem in Unit Disk Graph}
\author{\footnotesize Yishuo Shi$^1$ \quad Zhao Zhang$^1$\quad Ding-Zhu Du$^2$\\
    {\it\small $^1$ College of Mathematics and Computer Science, Zhejiang Normal University}\\
    {\it\small Jinhua, Zhejiang, 321004, China.}\\
    {\it\small $^2$ Department of Computer Science, University of Texas at Dallas}\\
    {\it\small Richardson, Texas, 75080, USA}}
\date{}
\maketitle

\begin{abstract}
In a wireless sensor network, the virtual backbone plays an important role. Due to accidental damage or energy depletion, it is desirable that the virtual backbone is fault-tolerant. A fault-tolerant virtual backbone can be modeled as a $k$-connected $m$-fold dominating set ($(k,m)$-CDS for short). In this paper, we present a constant approximation algorithm for the minimum weight $(k,m)$-CDS problem in unit disk graphs under the assumption that $k$ and $m$ are two fixed constants with $m\geq k$. Prior to this work, constant approximation algorithms are known for $k=1$ with weight and $2\leq k\leq 3$ without weight. Our result is the first constant approximation algorithm for the $(k,m)$-CDS problem with general $k,m$ and with weight. The performance ratio is $(\alpha+2.5k\rho)$ for $k\geq 3$ and $(\alpha+2.5\rho)$ for $k=2$, where $\alpha$ is the performance ratio for the minimum weight $m$-fold dominating set problem and $\rho$ is the performance ratio for the subset $k$-connected subgraph problem (both problems are known to have constant performance ratios.)

\end{abstract}

\section{Introduction}

Sensors have been applied for collecting data with many kinds of purposes,
such as safety protection system, environment monitoring, manufacture process management,
healthcare etc. Especially, sensors are important sources of big data.

Usually, a sensor has a small energy storage and a limited ability to process data. Hence, they have to share information through multihop transimissions. Every sensor has a transmission range. Any device located within
this range is able to receive data sent by the sensor. Such a mechanism enables sensors
to form a wireless communication network, which is often required to be connected.

Different from wired network, wireless network does not have a prefixed infrastructure. Instead, virtual backbone is used in the implementation of network operations,
such as broadcast, multicast, and unicast. A virtual backbone is a subset of network nodes and
every network operation request can be reduced to a corresponding operation on
the virtual backbone. Such a reduction plays an important role in saving storage and energy.
Besides, through such a reduction,
it is easy to design algorithm agreement and analyze the complexity of implementation. To achieve such a goal, there are two requirements for
every virtual backbone. The first is that every network node is adjacent with a virtual backbone node, so that it can communicate with the virtual backbone. Secondly, the set of backbone nodes
should be connected, so that information can be shared in the virtual backbone (and thus the whole network combining with the first requirement). To meet these requirements, a virtual backbone can be modeled as a connected dominating set (CDS).

Given a graph $G = (V, E)$, a {\em dominating set} (DS) is a subset $D$ of $V$ such
that each node $v$ in $V \setminus D$ is adjacent with at least one node of $D$.
A node in $D$ adjacent with $v$ is called a {\em dominator} of $v$. A dominating
set $D$ is a {\em connected dominating set} (CDS) if the subgraph of $G$ induced
by $D$, denoted as $G[D]$, is connected.

Because the energy of a sensor is supplied by
battery, if the battery is depleted, then the sensor can no longer work.
When sensors are deployed in hostile environment, charging or recharging
batteries is impossible. This brings us issues of fault-tolerance and
energy efficiency.

In practice, due to energy depletion and accidental damage, it is desirable
that the virtual backbone is fault-tolerant, in the sense that it can
still work when some backbone nodes fail. This consideration leads to
the concept of {\em minimum $k$-connected $m$-fold dominating
set problem} (abbreviated as $(k,m)$-MCDS), the goal of which is to find a minimum
node set $D$ such that every node in $V\setminus D$ has at least $m$-neighbors
in $D$ and the subgraph $G[D]$ is $k$-connected.

In many applications, different sensors have different significance which leads to different weights on different nodes. In such a setting, it is desirable to find a virtual backbone with minimum weight instead of minimum cardinality. Also, minimum weight sensor cover plays an important role for the study of maximum lifetime problem in a wireless sensor network \cite{Berman}. In fact, by a result of Garg and K\"{o}nemann \cite{Garg}, if the minimum weight sensor cover problem has a $\rho$-approximation, then the maximum lifetime problem will have a $(\rho+\varepsilon)$-approximation.

Motivated by these considerations, we study
the {\em minimum weight $k$-connected $m$-fold dominating
set problem} (abbreviated as $(k,m)$-MWCDS), the goal of which is to find a node set $D$ which is a $(k,m)$-CDS with the minimum total weight.

In this paper, we consider a homogeneous wireless sensor network, which means that all sensors have the same transmission range. We assume that
each sensor is equipped with an omnidirectional antenna with
transmission radius being one. Hence its transmission range
is a disk with radius one centered at this sensor.
Two sensors can communicate with each other if and only if they fall into
the transmission ranges of each other.
Therefore,
the communication network of such a wireless sensor network can be formulated as a {\em unit disk graph} in
which all nodes lie in the Euclidean plane and an edge exists between two nodes
if and only if their Euclidean distance is at most one.

In this paper, we present the first constant approximation algorithm for $(k,m)$-MWCDS in unit disk graphs, where $m,k$ are two fixed integers with $m\geq k\geq 2$.

\subsection{Related Work}

The concept of virtual backbone in a wireless sensor network was first proposed by Das and Bharghavan \cite{Das}.
This motivates the study of minimum connected dominating set (MCDS) in graphs,
especially in unit disk graphs.

It is known that MCDS is NP--hard even in unit disk graphs \cite{Clark}.
Moreover, MCDS cannot be polynomial-time approximated within a
factor of $(1-\varepsilon) \ln n$ for any $\varepsilon>0$ in general graphs \cite{Guha}.

For the MCDS problem, using partition method, Cheng {\it et al.} \cite{Cheng} gave a polynomial-time approximation
scheme (PTAS) in unit disk graphs.
By essentially the same method but different analysis, Zhang {\it et al.} \cite{Zhang}
obtained a PTAS in unit ball graph (a generalization of unit disk graph to higher dimensional space).
These are both centralized algorithms.

As to distributed algorithms for MCDS, Wan {\it et al.}
\cite{Wan} were the first to propose a constant approximation. Their algorithm has
performance ratio $2(mis(n)-1)$, where $mis(n)$ is the maximum number of
independent points (points with mutual distance greater than one) in the union of $n$
unit disks which induce a connected unit disk graph. It is easy to obtain an upper bound $4n+1$ for $mis(n)$. After a series of improvements
\cite{Funke,Gao,LiM,Wan1,WuW}, the current best upper bound for $mis(n)$ is
$3.399n+4.874$ \cite{LiM}.

The weighted version MWCDS is much more difficult.
The first constant approximation algorithm for MWCDS in unit disk graphs was proposed by Amb\"{u}hl {\it et al.} \cite{Ambuhl}. Their performance ratio is $89$, which consists of a $72$-approximation
for the minimum weight dominating set problem (MWDS) in unit disk graphs and a $17$-approximation
for the connecting part. The step stone for their $72$-approximation for MWDS is the
observation that the minimum weight {\em strip outside cover} problem (in which points
in a strip are to be covered by unit disks whose centers are outside of the strip) can
be solved in polynomial time by dynamic programming. Huang {\it et al.} \cite{Huang}
reduce the ratio for MWDS from $72$ to $(6+\varepsilon)$ by introducing a new
technique called ``double partition'', and reduce the ratio for the connecting part
from $17$ to $4$, making use of a minimum weight spanning tree in an auxiliary weighted
complete graph. Later, the ratio for MWDS was further improved to $(5+\varepsilon)$ by
Dai and Yu \cite{DaiYu}, to $(4+\varepsilon)$ by Zou {\it et al.} \cite{Zou} and
independently Erlebach and Mihal\'{a}k \cite{Erlebach}, to $(3+\varepsilon)$ by Willson and Zhang {\it et al.} \cite{Willson,ZhangTON}. Very recently, a PTAS was obtained by Li and Jin \cite{LiJ}. For the connecting part,
Zou {\it et al.} \cite{Zou2} gave a $2.5\rho_0$-approximation, where $\rho_0$ is the performance ratio for the minimum Steiner tree problem. Using currently best known ratio $\rho_0=1.39$ by Byrka {\it et al.} in paper \cite{Byrka}, the performance ratio for the connecting part is $3.475$.

For a better comprehensive study on MCDS and  MWCDS, the readers may refer to the book \cite{DuBookCDS} or chapters \cite{Blum,DuH} in Handbook of Combinatorial Optimization.

The study on fault-tolerant virtual backbone was initiated by Dai and Wu \cite{Dai}.
They presented three localized heuristic algorithms for $(k,k)$-MCDS, but no analysis on the
performance ratio was given. Wang {\it et al.} \cite{Wang} provided a $72$--approximation
for $(2,1)$--MCDS in unit disk graphs. Their strategy is to first find a connected dominating set and then increase its
connectivity to two by adding paths connecting different blocks (a block is a subgraph
without cut nodes and is maximal with respect to this property). The crucial point to
the performance ratio is that for a connected dominating set in a unit disk graph, there always exists a path between
different blocks with at most eight internal nodes. In the case $m\geq 2$, Shang {\it et al.}
\cite{Shang} gave an $\alpha_m$--approximation for $(2,m)$--MCDS in unit disk graphs, where
$\alpha_m=15+ \frac{15}{m}$ for $2\leq m\leq 5$ and $\alpha_m=21$ for $m>5$.
Their algorithm first finds a $(1,m)$-CDS and then augments the connectivity to two.
A key observation is that in the case $m\geq 2$, there always exists a path between
different blocks with at most two internal nodes. Wang {\it et al.} \cite{Wangw}
gave the first constant approximation for $(3,m)$--MCDS in unit disk graphs, which was further improved in \cite{WangWei2}.

There is also some  work on $(k,m)$--MCDS for general $k$ and $m$, in a unit disk graph
\cite{Li,Wu} or even in a disk graph \cite{Thai} (which models a heterogeneous wireless
sensor network). However, for $k\geq 4$, whether there exists a constant approximation algorithm for $(k,m)$--MCDS on unit disk graph is still unknown.

For $(k,m)$-MCDS in general graphs, Zhou {\it et al.} \cite{Zhou} presented a
$\beta_1$-approximation for $(1,m)$-MCDS, where $\beta_1=2+H(\bigtriangleup+m-2)$ and $H(\gamma)=\sum_{i=1}^\gamma 1/i$ is the Harmonic number. Shi {\it et al.} \cite{Shi} presented
a $\beta_2=(\beta_1+2(1+\ln\beta_1))$-approximation for $(2,m)$-MCDS with $m\geq 2$. When applied to unit disk graphs, this algorithm reduces previous ratio in \cite{Shang} by more than half. Zhang {\it et al.} \cite{ZhangInfocom} obtained a $(\beta_2+8+2\ln(2\beta_2-6))$-approximation for $(3,m)$-CDS with $m\geq 3$. When applied to unit disk graphs, this algorithm reduces previous ratio in \cite{WangWei2} from more than 62 to less than 27. Since $H(\gamma)\approx\ln\gamma+0.577$, these algorithms have performance ratio $\ln \Delta+o(\ln\Delta)$ for general graphs. In view of the
inapproximability of this problem \cite{Guha}, these ratios are asymptotically best possible. For $k>1$, there is no
work on weighted version of $(k,m)$-MCDS in general graphs.

\subsection{Our Contributions}

Recall that for $k\geq 4$, whether $(k,m)$--MCDS on unit disk graphs has a constant approximation is still unknown, even in the simpler case without weight. In this paper, we answer this open problem confirmatively by presenting a constant approximation algorithm for $(k,m)$--MWCDS (with weight), where $m,k$ are two fixed integers with $m\geq k$. The algorithm is executed in two steps: First it finds an $m$--fold dominating set $D$. Then, it computes a $k$--connected subgraph $F$ containing $D$. To realize the second step, we design an approximation algorithm for the minimum node-weighted $k$--connected Steiner
network problem (MNW$k$CSN) in which the terminal set is an $m$-fold dominating set with $m\geq k$. By implementing an approximation algorithm for the subset $k$--connected subgraph problem (SkCS), we prove that for unit disk graphs, our algorithm for the special MNW$k$CSN problem has performance ratio $2.5k\rho$ when $k\geq 3$ and $2.5\rho$ when $k=2$, where $\rho$ is the performance ratio for SkCS.
Combining these two steps together, our algorithm for $(k,m)$-MWCDS has performance ratio
$(\alpha+2.5k\rho)$ when $k\geq 3$ and $(\alpha+2.5\rho)$ when $k=2$, where $\alpha$ is the performance ratio for the minimum weight $m$--fold dominating set problem.

In \cite{Fukunage}, Fukunage obtained an $O(1)$--approximation for the minimum weight $m$-fold dominating set problem which is valid for any positive integer $m$. As to SkCS,
the best known ratio is $\rho=O(k^{2}\log k)$ due to Nutov \cite{Nutov} when $k\geq 3$ and $\rho=2$ when $k=2$ due to Fleischer \cite{Fleischer}. So our algorithm has constant
performance ratio for fixed integers $m$ and $k$ with $m\geq k$. Although our analysis makes use of a lot of geometry, the execution of our algorithm does not need a geometric representation of the unit disk graph on the plane.

The remainder of this paper is organized as follows. Section \ref{sec2} presents the
algorithm together with its performance ratio analysis. Section \ref{sec3} concludes
the paper and proposes some future work.

\section{Approximation algorithm for $(k,m)$--MWCDS}\label{sec2}

\begin{definition}[$k$--connected $m$--fold dominating set (($k,m$)--CDS)]\label{def1}
Given a graph $G=(V, E)$, two positive integers $k$ and $m$,
 and a cost function $c: V \rightarrow \mathbb R^{+}$,
 a node subset $D\subseteq V$ is a $(k,m)$--CDS if

$(a)$ every node in $V \setminus D$ is adjacent with at least $m$ nodes of $D$, and

$(b)$ the subgraph of $G$ induced by $D$ is $k$--connected.

\noindent The minimum weight $k$--connected $m$--fold dominating set problem, abbreviated as $(k,m)$--MWCDS, is to find a $(k,m)$--CDS $D$ with $c(D)=\sum_{v\in D}c(v)$  minimized. In particular, if $c \equiv 1$, then we have the unweighted  minimum $k$--connected $m$--fold dominating set problem, abbreviated as $(k,m)$--MCDS.
\end{definition}

In this section, we shall design a constant approximation algorithm for $(k,m)$--MWCDS, where $m\geq k$. In Subsection \ref{subsec3.1}, a geometric property for unit disk graph is obtained, showing that every $k$--connected unit disk graph has a $k$--connected spanning subgraph whose maximum degree is upper bounded by a constant (related to $k$). In Subsection \ref{subsec3.2}, we present an approximation algorithm for the minimum node-weighted $k$-connected Steiner network problem in unit disk graphs in which the terminal set is an $m$-fold dominating set with $m\geq k$. The algorithm for $(k,m)$-MWCDS is presented in Subsection \ref{subsec3.3}.

\subsection{$k$--Connected Spanning Subgraph of Unit Disk Graph}\label{subsec3.1}
In the following, we always assume that the unit disk graph $G$ is embedded on the plane, and the {\em length} of an edge $uv$ is the Euclidean length of line segment $uv$, denoted as $\|uv\|$. The length of a subgraph $F$ of $G$ is $len(F)=\sum_{e\in E(F)}\|e\|$.

In \cite{Holberg}, Holberg gave a kind of decomposition of $k$--connected graphs. We describe it in the following, using language which is consistent with this paper.

Let $G=(V,E)$ be a simple graph. For a node subset $S\subseteq V$, the subgraph of $G$ induced by $S$ is denoted as $G[S]$. We shall use $K_S$ to denote a complete graph on node set $S$. For a connected graph $G$, a node set $S$ is a {\em separator} of $G$ if $G-S$ is disconnected. A separator of size $k$ is called a {\em $k$-separator}. Since parallel edges and loops do not affect the vertex connectivity, we always assume that the graph under consideration is simple. Hence when a multi-graph is created by some operation, redundant edges are removed to keep the graph to be simple.

\begin{definition}[$S$-component and marked $S$-component]
{\rm Let $G$ be a $k$-connected graph, $S$ be a $k$-separator of $G$, and $C$ be a connected component of $G-S$. The subgraph $G[C\cup S]$ is called an {\em $S$-component} of $G$, and the graph $G[C\cup S]\cup K_{S}$ is called a {\em marked $S$-component}.} where $K_{S}$ is a complete graph of $|S|$ vertices.
\end{definition}

\begin{figure*}[htp]
\begin{center}
\hskip -0.5cm
\begin{picture}(80,140)
\put(10,10){\circle*{5}}\put(70,10){\circle*{5}}\put(40,46){\circle*{5}}
\put(10,60){\circle*{5}}\put(70,60){\circle*{5}}
\put(10,90){\circle*{5}}\put(70,90){\circle*{5}}\put(40,110){\circle*{5}}
\qbezier(10,10)(40,10)(70,10)\qbezier(10,90)(40,90)(70,90)
\qbezier(10,10)(10,50)(10,90)\qbezier(70,10)(70,50)(70,90)
\qbezier(10,60)(40,35)(70,10)\qbezier(70,60)(55,85)(40,110)
\qbezier(40,46)(25,53)(10,60)\qbezier(40,46)(55,53)(70,60)
\qbezier(10,60)(40,75)(70,90)\qbezier(10,90)(49,75)(70,60)
\qbezier(40,110)(25,100)(10,90)\qbezier(40,110)(55,100)(70,90)
\put(-3,57){$u_1$}\put(74,57){$u_2$}\put(74,7){$u_3$}\put(-3,7){$u_4$}
\put(36,52){$u_5$}\put(-3,87){$u_6$}\put(74,87){$u_7$}\put(38,115){$u_8$}
\put(36,-5){(a)}
\put(90,63){\vector(4,3){40}}\put(90,60){\vector(1,0){40}}\put(90,57){\vector(4,-3){40}}
\end{picture}
\hskip 2cm\begin{picture}(80,140)
\put(10,10){\circle*{5}}\put(70,10){\circle*{5}}
\put(10,40){\circle*{5}}\put(70,40){\circle*{5}}
\qbezier(10,10)(40,10)(70,10)\qbezier(10,40)(40,25)(70,10)
\qbezier(10,10)(10,25)(10,40)\qbezier(70,10)(70,25)(70,40)
{\linethickness{0.25mm}\qbezier[13](10,40)(40,40)(70,40)}
\put(-3,37){$u_1$}\put(74,37){$u_2$}\put(74,7){$u_3$}\put(-3,7){$u_4$}

\put(10,66){\circle*{5}}\put(70,66){\circle*{5}}\put(40,50){\circle*{5}}
\qbezier(40,50)(25,58)(10,66)\qbezier(40,50)(55,58)(70,66)
{\linethickness{0.25mm}\qbezier[13](10,66)(40,66)(70,66)}
\put(-3,63){$u_1$}\put(63,57){$u_2$}\put(36,55){$u_5$}\put(85,58){$B_2$}

\put(10,80){\circle*{5}}\put(70,80){\circle*{5}}
\put(10,104){\circle*{5}}\put(70,104){\circle*{5}}\put(40,120){\circle*{5}}
\qbezier(10,80)(10,92)(10,104)\qbezier(70,80)(70,92)(70,104)
\qbezier(10,80)(40,92)(70,104)\qbezier(70,80)(40,92)(10,104)
\qbezier(40,120)(25,112)(10,104)\qbezier(40,120)(55,112)(70,104)
\qbezier(40,120)(55,100)(70,80)\qbezier(10,104)(40,104)(70,104)
{\linethickness{0.25mm}\qbezier[13](10,80)(40,80)(70,80)}
\put(-3,77){$u_1$}\put(74,77){$u_2$}\put(-3,101){$u_6$}\put(74,101){$u_7$}\put(44,120){$u_8$}
\put(85,90){$B_1$}

\put(36,-5){(b)}
\put(90,15){\vector(1,0){40}}\put(90,18){\vector(3,2){40}}
\end{picture}
\hskip 2cm\begin{picture}(80,140)
\put(10,10){\circle*{5}}\put(60,10){\circle*{5}}\put(10,30){\circle*{5}}
\qbezier(10,10)(35,10)(60,10)\qbezier(10,30)(35,20)(60,10)\qbezier(10,10)(10,20)(10,30)
\put(-3,27){$u_1$}\put(64,7){$u_3$}\put(-3,7){$u_4$}\put(75,12){$B_4$}

\put(60,40){\circle*{5}}\put(10,60){\circle*{5}}\put(60,60){\circle*{5}}
\qbezier(10,60)(35,50)(60,40)\qbezier(60,40)(60,50)(60,60)
{\linethickness{0.25mm}\qbezier[13](10,60)(35,60)(60,60)}
\put(-3,57){$u_1$}\put(64,57){$u_2$}\put(64,27){$u_3$}\put(75,42){$B_3$}

\put(36,-5){(c)}
\end{picture}
\end{center}
\vskip 0.3cm
\caption{An illustration of marked components and the $2$-block decomposition of a 2-connected graph.}\label{fig16-3-30-1}
\end{figure*}
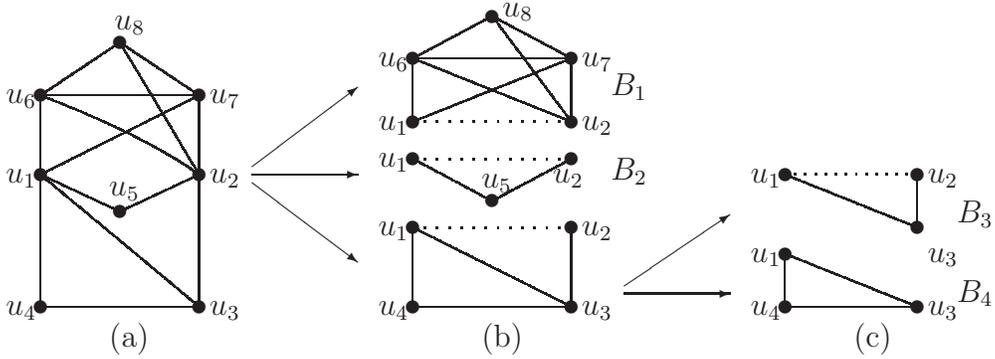

Notice that the difference between an $S$-component and its corresponding marked $S$-component is that the marked $S$-component may have more edges because of the addition of $K_S$, which are called {\em virtual edges}.

The above concepts are illustrated by Fig.\ref{fig16-3-30-1}. The graph $G$ in Fig.\ref{fig16-3-30-1}(a) has connectivity 2. Node set $S=\{u_1,u_2\}$ is a 2-separator of $G$. The marked $S$-components are depicted in Fig.\ref{fig16-3-30-1}(b). Notice that $u_1u_2$ is not an edge in $G$. While in those marked components, virtual edges (indicated by the dashed lines) are added to join $u_1$ and $u_2$. The bottom marked $S$-component in Fig. \ref{fig16-3-30-1}(b) has a 2-separator $S'=\{u_1,u_3\}$. Marked $S'$-components are depicted in Fig.\ref{fig16-3-30-1}(c). Notice that $u_1u_3$ is already an edge. So in this decomposition, no virtual edges are needed.

The reason why virtual edges are used can be seen from the following lemma. Its proof is easy. For better understanding of the decomposed structure, we include its proof here.

\begin{lemma}\label{lem15-3-30-1}
Let $G$ be a $k$-connected graph and $S$ be a $k$-separator of $G$. Then any marked $S$-component of $G$ is also $k$-connected.
\end{lemma}
\begin{proof}
Consider a marked $S$-component $G'$ of $G$. Suppose $S'$ is a separator of $G'$ and $|S'|<k$. Let $G_1'$ be the component of $G'-S'$ which contains $K_S\setminus S'$ (notice that $K_S\setminus S'\neq\emptyset$ since $|S|\geq k$ and $|S'|<k$, and $K_S\setminus S'$ induces a connected subgraph of $G'$), and let $G_2'$ be another component of $G'-S'$. Then we see that any path in $G$ (not only in $G'$) which connects $G_1'$ and $G_2'$ must go through $S'$, and thus $S'$ is also a separator of $G$, contradicting that $G$ is $k$-connected.
\end{proof}

Suppose $G$ is a $k$-connected graph which has a $k$-separator $S$. By Lemma \ref{lem15-3-30-1}, $G$ can be decomposed into several marked $S$-components which are also $k$-connected. If any one of these marked components, say $G'$, also has a $k$-separator $S'$, then $G'$ can be further decomposed into several marked $S'$-components. Such an operation can be recursively executed until no marked component has a $k$-separator. See Fig.\ref{fig16-3-30-1} for an illustration. The graph $G$ in Fig.\ref{fig16-3-30-1}(a) can be decomposed through the 2-separator $S=\{u_1,u_2\}$ into three marked $S$-components in Fig.\ref{fig16-3-30-1}(b). The bottom marked $S$-component is further decomposed through the 2-separator $S'=\{u_1,u_3\}$ into two marked $S'$-components in Fig.\ref{fig16-3-30-1}(c). The upper marked $S$-component in Fig.\ref{fig16-3-30-1}(b) is 3-connected, the middle marked $S$-component in Fig.\ref{fig16-3-30-1}(b) and the two marked $S'$-components in Fig.\ref{fig16-3-30-1}(c) are $K_3$'s. Since no one of them contains a 2-separator, the decomposition halts. Notice that any $k$-connected graph without $k$-separators is either a $K_{k+1}$ or a $(k+1)$-connected graph. So, in the final decomposition, there are two types of marked components, $K_{k+1}$ and $(k+1)$-connected marked component. For convenience of statement, we call these marked components {\em $k$-blocks}. For example, the graph in Fig.\ref{fig16-3-30-1} is decomposed into four $2$-blocks: $B_1,B_2,B_3,B_4$.

The original graph $G$ can be viewed as pasting these $k$-blocks through those $k$-separators used in the decomposition and ignoring those virtual edges. From such a point of view, $G$ has a {\em tree-like} structure (see Fig.\ref{fig1} for an illustration). To be more concrete, let $B_k(G)$ be a bipartite graph with bipartition $(X,Y)$, where every vertex in $X$ corresponds to a $k$-block in the final decomposition and every vertex in $Y$ corresponds to a $k$-separator used in the decomposition. Vertex $x\in X$ is adjacent with vertex $y\in Y$ in $B_k(G)$ if and only if the $k$-separator corresponding to $y$ is contained in the $k$-block corresponding to $x$. It can be seen that $B_k(G)$ is a tree (see Fig.\ref{fig1}(b)). We call $B_k(G)$ the {\em $k$-block tree of $G$}. Those $k$-blocks which correspond to leaves of $B_k(G)$ are called {\em leaf $k$-blocks}.

\begin{figure}[!htbp]
\begin{center}
\begin{picture}(50,130)
\qbezier(10,10)(-10,30)(10,50)\qbezier(10,10)(30,30)(10,50)
\qbezier(10,50)(-10,70)(10,90)\qbezier(10,50)(30,70)(10,90)
\qbezier(10,90)(-10,110)(10,130)\qbezier(10,90)(30,110)(10,130)
\qbezier(10,90)(30,70)(50,90)\qbezier(10,90)(30,110)(50,90)
\put(5,105){$B_1$}\put(25,87){$B_2$}\put(5,65){$B_3$}\put(5,25){$B_4$}
\put(-5,85){$S$}\put(-5,45){$S'$}
\put(3,-5){(a)}
\end{picture}
\hskip 1cm\begin{picture}(40,130)
\put(10,10){\circle*{5}}\put(10,40){\circle*{5}}\put(10,70){\circle*{5}}
\put(10,100){\circle*{5}}\put(40,100){\circle*{5}}\put(10,130){\circle*{5}}
\qbezier(10,10)(10,70)(10,130)\qbezier(10,100)(25,100)(40,100)
\put(-5,128){$B_1$}\put(-5,98){$S$}\put(-5,68){$B_3$}
\put(-5,38){$S'$}\put(-5,8){$B_4$}\put(35,88){$B_2$}
\put(3,-5){(b)}
\end{picture}
\end{center}
\vskip 0.3cm
\caption{$(a)$ The $2$-block structure of the graph $G$ in Fig.\ref{fig16-3-30-1}. Each ellipse represents a $2$-block. There are two $2$-separators in this graph, namely $S$ and $S'$. \  $(b)$ The $2$-block tree of $G$. There are three leaf $2$-blocks in $G$, namely $B_1,B_2$ and $B_4$, which are leaves of the $2$-block tree.}\label{fig1}
\end{figure}
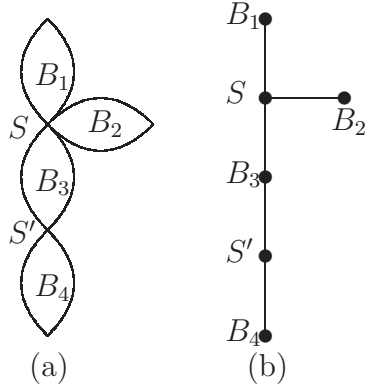

The neighbor set of a node $u$ in graph $G$ is denoted as $N_G(u)$. Its degree $d_G(u)=|N_G(u)|$. For a $k$--connected graph $G$, a minimum length $k$--connected spanning subgraph of $G$ is abbreviated as $k$-MSS. Notice that a $k$-MSS is an {\em edge induced subgraph}, not a node induced subgraph.

\begin{lemma}\label{lemma1}
Let $F$ be a $k$--MSS of a $k$--connected unit disk graph $G$. Then the maximum degree of $F$ is no more than $5k$.
\end{lemma}

\begin{proof}
First, we have the following properties.

$(a)$ For any edge $uv$, $F-uv$ has connectivity $k-1$ because $F$ is minimal with respect to $k$-connectivity.

The $(k-1)$-block tree of $F-uv$ is a path. To see this, denote by $B_u$ and $B_v$ the two $(k-1)$-blocks of $F-uv$ containing $u$ and $v$, respectively. Let $P_{uv}$ be the unique path on the $(k-1)$-block tree $B_{k-1}(G)$ connecting $B_u$ and $B_v$. Adding edge $uv$ back to $F-uv$ will merge those $(k-1)$-blocks of $F-uv$ on $P_{uv}$ into a larger $(k-1)$-block (see Fig.\ref{fig15-7-23-1}(a)). So, if $B_{k-1}(G)$ is not a path, then $F-uv$ has a leaf block outside of $P_{uv}$, which cannot be merged (see Fig.\ref{fig15-7-23-1}(b)), contradicting that $F=(F-uv)+uv$ is $k$-connected.

For the same reason, nodes $u,v$ must belong to the two leaf $(k-1)$-blocks of $F- uv$, respectively. Furthermore, let $S_u$ be the unique $(k-1)$-separator contained in $B_u$ and let $S_v$ be the unique $(k-1)$-separator contained in $B_v$, respectively. We must have $u\not\in S_u$ and $v\not\in S_v$ (Fig.\ref{fig15-7-23-1}(c))

$(b)$ $N_F(u)\setminus\{v\}\subseteq B_u$ and $N_F(v)\setminus \{u\}\subseteq B_v$.

$(c)$ As a consequence of $(b)$, if $d_F(u)\geq k+1$, then $|B_u|\geq k+1$, and thus $B_u$ is $k$-connected.

\vskip 0.1cm
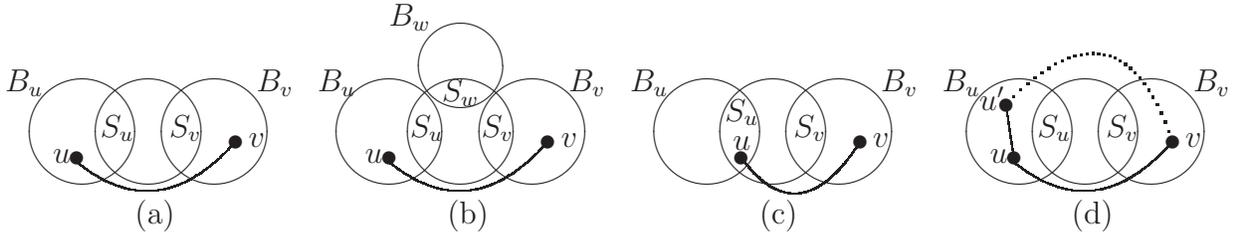
\begin{figure}[!hbtp]
\begin{center}
\hskip -0.cm \begin{picture}(100,60)
\put(20,30){\circle{40}}\put(45,30){\circle{40}}\put(70,30){\circle{40}}
\put(18,20){\circle*{5}}\put(78,26){\circle*{5}}
\qbezier(16,20)(48,-8)(78,26)
\put(9,18){$u$}\put(83,24){$v$}\put(-9,45){$B_u$}\put(86,45){$B_v$}\put(27,28){$S_u$}\put(53,28){$S_v$}
\put(40,-5){(a)}
\end{picture}
\hskip 0.5cm
\begin{picture}(100,80)
\put(18,30){\circle{40}}\put(45,30){\circle{40}}\put(72,30){\circle{40}}\put(45,55){\circle{30}}
\put(18,20){\circle*{5}}\put(78,26){\circle*{5}}
\qbezier(16,20)(48,-8)(78,26)
\put(9,18){$u$}\put(83,24){$v$}\put(-9,45){$B_u$}\put(86,45){$B_v$}\put(18,70){$B_w$}\put(26,28){$S_u$}\put(53,28){$S_v$}\put(38,41.5){$S_w$}
\put(40,-5){(b)}
\end{picture}
\hskip 0.5cm \begin{picture}(100,50)
\put(20,30){\circle{40}}\put(45,30){\circle{40}}\put(70,30){\circle{40}}
\put(33,20){\circle*{5}}\put(78,26){\circle*{5}}
\qbezier(33,20)(55,-10)(78,26)
\put(30,23){$u$}\put(83,24){$v$}\put(-9,45){$B_u$}\put(86,45){$B_v$}\put(27,34){$S_u$}\put(53,28){$S_v$}
\put(40,-5){(c)}
\end{picture}
\hskip 0.5cm
\begin{picture}(100,60)
\put(20,30){\circle{40}}\put(45,30){\circle{40}}\put(70,30){\circle{40}}
\put(18,20){\circle*{5}}\put(15,40){\circle*{5}}\put(78,26){\circle*{5}}
\qbezier(16,20)(48,-8)(78,26)\qbezier(18,20)(16.5,30)(15,40)
{\linethickness{0.26mm}\qbezier[30](15,40)(58,85)(78,26)}
\put(9,18){$u$}\put(5,38){$u'$}\put(83,24){$v$}\put(-9,45){$B_u$}\put(86,45){$B_v$}\put(27,28){$S_u$}\put(53,28){$S_v$}
\put(40,-5){(d)}
\end{picture}
\end{center}
\vskip 0.3cm
\caption{$(a)$ An illustration for the structure of $F-uv$. $(b)$ An illustration of why the $(k-1)$-block structure of $F-uv$ cannot have more than two leaf $(k-1)$-blocks: adding $uv$ to $F-uv$ cannot eliminate the $(k-1)$-separator $S_w$. $(c)$ An illustration of why $u,v$ must belong to leaf $(k-1)$-blocks and why they cannot belong to corresponding $(k-1)$-separators: adding $uv$ to $F-uv$ cannot eliminate the $(k-1)$-separator $S_u$. $(d)$ An illustration for the proof of the claim in Lemma \ref{lemma1}.} \label{fig15-7-23-1}
\end{figure}

{\bf Claim.} If $d_F(u)\geq k+1$, then any node $u'\in N_F(u)\setminus (S_u\cup\{v\})$ has $\angle u'uv\geq \pi/3$.

Since $d_F(u)\geq k+1$, we have $N_F(u)\setminus (S_u\cup\{v\})\neq\emptyset$. Suppose the claim is not true. Let $u'$ be a node in $N_F(u)\setminus (S_u\cup\{v\})$ with $\angle u'uv<\pi/3$. Notice that $u'v\not\in E(F)$ by the $(k-1)$-block structure of $F-uv$ (see Fig.\ref{fig15-7-23-1}$(d)$). First consider the case that $\|uu'\|\leq\|uv\|$. In this case $\|u'v\|<\|uv\|$ and thus $u'v\in E(G)\setminus E(F)$. Since $u'\not\in S_u$ and $v\not\in S_v$, subgraph $F'=F-uv+u'v$ is also a $k$-connected spanning subgraph of $G$ (see Fig.\ref{fig15-7-23-1}$(d)$). However, $len(F')=len(F)-\|uv\|+\|u'v\|<len(F)$, contradicting the minimality of $F$. Next, we consider the case that $\|uv\|<\|uu'\|$. In this case, $\|u'v\|<\|uu'\|$, and a contradiction will follow as long as we can prove that $F'=F-uu'+u'v$ is $k$-connected. For this purpose, notice that $B_u$ is $k$-connected by property $(c)$. So, $B_u-uu'$ is $(k-1)$-connected. Furthermore, if $B_u-uu'$ is not $k$-connected, the $(k-1)$-block tree of $B_u-uu'$ must be a path, and $u,u'$ are distributed in the two leaf $(k-1)$-blocks of $B_u-uu'$, avoiding the corresponding $(k-1)$-separators (this is similar to the proof of property $(a)$, since otherwise adding $uu'$ back will not result in a $k$-connected marked component). As a consequence, if $B_u-uu'$ is $k$-connected, then the $(k-1)$-block tree of $F-uv-uu'$ is still a path. If $B_u-uu'$ is not $k$-connected, then the $(k-1)$-block tree of $F-uv-uu'$ has the shape in Fig.\ref{fig16-3-31-1}. In any case, $F'=(F-uv-uu')+uv+u'v$ is a $k$-connected spanning subgraph of $G$ with shorter length than $F$, a contradiction. The claim is proved.

\begin{figure}[!hbtp]\setlength{\unitlength}{0.3mm}
\begin{center}
\begin{picture}(110,100)
\put(67,43){\circle{40}}\put(90,43){\circle{40}}
\put(28,70){\circle*{5}}\put(10,24){\circle*{5}}\put(90,33){\circle*{5}}
{\linethickness{0.25mm}\qbezier[21](28,70)(19,47)(10,24)}
\put(32,68){$u'$}\put(13,17){$u$}\put(93,30){$v$}\put(1,6){$B'_u$}\put(-5,73){$B'_{u'}$}\put(35,35){$B'_x$}\put(90,65){$B_v$}
\put(48,40){$S_u$}\put(72,40){$S_v$}\put(30,50){$S$}\put(20,27){$S'$}

\qbezier(60.0,45.0)(59.7,50.0)(58.8,54.9)
\qbezier(58.8,54.9)(57.3,59.5)(55.2,63.9)
\qbezier(55.2,63.9)(52.7,67.9)(49.6,71.5)
\qbezier(49.6,71.5)(46.2,74.4)(42.5,76.8)
\qbezier(42.5,76.8)(38.5,78.6)(34.3,79.6)
\qbezier(34.3,79.6)(30.0,80.0)(25.7,79.6)
\qbezier(25.7,79.6)(21.5,78.6)(17.5,76.8)
\qbezier(17.5,76.8)(13.8,74.4)(10.4,71.5)
\qbezier(10.4,71.5)(7.3,67.9)(4.8,63.9)
\qbezier(4.8,63.9)(2.7,59.5)(1.2,54.9)
\qbezier(1.2,54.9)(0.3,50.0)(0.0,45.0)
\qbezier(0.0,45.0)(0.3,40.0)(1.2,35.1)
\qbezier(1.2,35.1)(2.7,30.5)(4.8,26.1)
\qbezier(4.8,26.1)(7.3,22.1)(10.4,18.5)
\qbezier(10.4,18.5)(13.8,15.6)(17.5,13.2)
\qbezier(17.5,13.2)(21.5,11.4)(25.7,10.4)
\qbezier(25.7,10.4)(30.0,10.0)(34.3,10.4)
\qbezier(34.3,10.4)(38.5,11.4)(42.5,13.2)
\qbezier(42.5,13.2)(46.2,15.6)(49.6,18.5)
\qbezier(49.6,18.5)(52.7,22.1)(55.2,26.1)
\qbezier(55.2,26.1)(57.3,30.5)(58.8,35.1)
\qbezier(58.8,35.1)(59.7,40.0)(60.0,45.0)

\qbezier(57.5,58.5)(56.5,60.3)(54.6,61.9)
\qbezier(54.6,61.9)(51.6,63.1)(47.8,63.9)
\qbezier(47.8,63.9)(43.2,64.2)(38.1,64.2)
\qbezier(38.1,64.2)(32.7,63.7)(27.1,62.8)
\qbezier(27.1,62.8)(21.6,61.5)(16.3,60.0)
\qbezier(16.3,60.0)(11.6,58.1)(7.5,56.0)
\qbezier(7.5,56.0)(4.3,53.9)(2.0,51.6)
\qbezier(2.0,51.6)(0.7,49.5)(0.5,47.5)
\qbezier(0.5,47.5)(1.5,45.7)(3.4,44.1)
\qbezier(3.4,44.1)(6.4,42.9)(10.2,42.1)
\qbezier(10.2,42.1)(14.8,41.8)(19.9,41.8)
\qbezier(19.9,41.8)(25.3,42.3)(30.9,43.2)
\qbezier(30.9,43.2)(36.4,44.5)(41.7,46.0)
\qbezier(41.7,46.0)(46.4,47.9)(50.5,50.0)
\qbezier(50.5,50.0)(53.7,52.1)(56.0,54.4)
\qbezier(56.0,54.4)(57.3,56.5)(57.5,58.5)

\qbezier(0.9,50.2)(0.2,48.4)(0.3,45.9)
\qbezier(0.3,45.9)(1.2,43.0)(3.0,39.6)
\qbezier(3.0,39.6)(5.6,36.0)(8.9,32.3)
\qbezier(8.9,32.3)(12.7,28.6)(16.8,25.1)
\qbezier(16.8,25.1)(21.3,21.9)(25.8,19.1)
\qbezier(25.8,19.1)(30.1,16.9)(34.2,15.2)
\qbezier(34.2,15.2)(37.9,14.3)(41.0,14.1)
\qbezier(41.0,14.1)(43.4,14.6)(45.1,15.8)
\qbezier(45.1,15.8)(45.8,17.6)(45.7,20.1)
\qbezier(45.7,20.1)(44.8,23.0)(43.0,26.4)
\qbezier(43.0,26.4)(40.4,30.0)(37.1,33.7)
\qbezier(37.1,33.7)(33.3,37.4)(29.2,40.9)
\qbezier(29.2,40.9)(24.7,44.1)(20.2,46.9)
\qbezier(20.2,46.9)(15.9,49.1)(11.8,50.8)
\qbezier(11.8,50.8)(8.1,51.7)(5.0,51.9)
\qbezier(5.0,51.9)(2.6,51.4)(0.9,50.2)
\put(50,-5){$(a)$}
\end{picture}
\hskip 1cm\begin{picture}(120,100)
\put(10,40){\circle*{5}}\put(30,40){\circle*{5}}\put(50,40){\circle*{5}}\put(70,40){\circle*{5}}\put(90,40){\circle*{5}}
\put(10,0){\circle*{5}}\put(10,20){\circle*{5}}\put(10,60){\circle*{5}}\put(10,80){\circle*{5}}
\qbezier(10,40)(50,40)(90,40)\qbezier(10,0)(10,40)(10,80)
{\linethickness{0.25mm}\qbezier[33](10,80)(90,70)(90,40)\qbezier[33](10,0)(90,10)(90,40)}
\put(-5,-4){$B'_u$}\put(-5,16){$S'$}\put(-5,36){$B'_x$}\put(-5,56){$S$}\put(-5,75){$B'_{u'}$}
\put(26,28){$S_u$}\put(66,28){$S_v$}\put(94,36){$B_v$}
\put(50,-5){$(a')$}
\end{picture}
\begin{picture}(110,110)
\put(65,43){\circle{40}}\put(90,43){\circle{40}}
\put(20,68){\circle*{5}}\put(20,20){\circle*{5}}\put(90,33){\circle*{5}}
\qbezier(45,43)(15,43)(10,65)\qbezier(45,43)(15,43)(10,22)
\qbezier(10,65)(10,100)(52,59)\qbezier(10,22)(10,-14)(52,27)
\qbezier(52,59)(70,43)(52,27)

\qbezier(43.6,42.8)(42.4,43.0)(40.9,42.5)
\qbezier(40.9,42.5)(39.3,41.5)(37.5,39.8)
\qbezier(37.5,39.8)(35.7,37.7)(33.9,35.1)
\qbezier(33.9,35.1)(32.2,32.2)(30.6,29.0)
\qbezier(30.6,29.0)(29.2,25.8)(28.1,22.6)
\qbezier(28.1,22.6)(27.3,19.5)(26.8,16.8)
\qbezier(26.8,16.8)(26.7,14.3)(26.9,12.4)
\qbezier(26.9,12.4)(27.5,11.0)(28.4,10.2)
\qbezier(28.4,10.2)(29.6,10.0)(31.1,10.5)
\qbezier(31.1,10.5)(32.7,11.5)(34.5,13.2)
\qbezier(34.5,13.2)(36.3,15.3)(38.1,17.9)
\qbezier(38.1,17.9)(39.8,20.8)(41.4,24.0)
\qbezier(41.4,24.0)(42.8,27.2)(43.9,30.4)
\qbezier(43.9,30.4)(44.7,33.5)(45.2,36.2)
\qbezier(45.2,36.2)(45.3,38.7)(45.1,40.6)
\qbezier(45.1,40.6)(44.5,42.0)(43.6,42.8)

{\linethickness{0.25mm}\qbezier[21](20,68)(20,44)(20,20)}
\put(16,71){$u'$}\put(18,12){$u$}\put(93,30){$v$}
\put(-10,15){$B'_u$}\put(-10,65){$B'_{u'}$}\put(50,13){$B'_x$}\put(95,65){$B_v$}
\put(47,40){$S_u$}\put(72,40){$S_v$}\put(34,24){$S$}
\put(50,-5){$(b)$}
\end{picture}
\hskip 1cm\begin{picture}(110,110)
\put(10,0){\circle*{5}}\put(10,20){\circle*{5}}\put(10,40){\circle*{5}}\put(10,60){\circle*{5}}\put(10,80){\circle*{5}}
\put(30,60){\circle*{5}}\put(50,60){\circle*{5}}\put(70,60){\circle*{5}}
\qbezier(10,0)(10,40)(10,80)\qbezier(10,60)(40,60)(70,60)
{\linethickness{0.25mm}\qbezier[19](10,80)(50,90)(70,60)\qbezier[31](10,0)(70,0)(70,60) }
\put(-5,-4){$B'_u$}\put(-5,16){$S'$}\put(-5,36){$B'_x$}\put(-5,56){$S_u$}\put(-5,75){$B'_{u'}$}
\put(46,48){$S_v$}\put(74,56){$B_v$}
\put(50,-5){$(b')$}
\end{picture}
\end{center}
\caption{An illustration of the $(k-1)$-block structure of $F-uv-uu'$. $(a)$ illustrates the case when $S_u$ is not a $(k-1)$-separator of $B_u-uu'$ and $(b)$ illustrates the case when $S_u$ is a $(k-1)$-separator of $B_u-uu'$. $(a')$ and $(b')$ are corresponding $(k-1)$-block trees. Adding edges $uv$ and $u'v$ (indicated by dashed lines) results in a $k$-connected graph.}\label{fig16-3-31-1}
\end{figure}
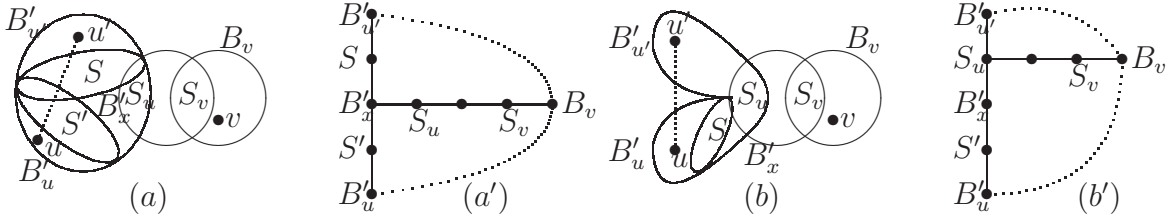

Now, consider a node $u$ with $d_F(u)\geq k+1$ (see Fig.\ref{fig15-8-5-1}). Let $v_1$ be a node in $N_F(u)$, and let $A_1$ be the set of nodes $u_1\in N_F(u)\setminus\{v_1\}$ with $\angle u_1uv_1<\pi/3$. By the above claim, $A_1\subseteq S_u$ and thus $|A_1|\leq k-1$. Let $v_2$ be the first node in $N_F(u)$ which has $\angle v_1uv_2\geq \pi/3$ (where `first' is counted clockwise), and let $A_2$ be the set of nodes $u_2\in N_F(u)\setminus\{v_2\}$ with $\angle u_2uv_2<\pi/3$. Similar to the above, $|A_2|\leq k-1$. Continuing this procedure, we obtain a sequence of nodes $v_1,\ldots,v_5$ and a sequence of sets $A_1,\ldots,A_5$ such that for each $i=1,\ldots,4$, $v_{i+1}$ is the first node in $N_F(u)$ with $\angle v_iuv_{i+1}\geq \pi/3$ and $|A_i|\leq k-1$. Clearly, $N_G(u)\subseteq\big(\bigcup_{i=1}^tA_i\big)\cup\{v_i\}_{i=1}^5$. Hence $|N_F(u)|\leq 5k$. The lemma is proved.
\end{proof}

\vskip 0.1cm
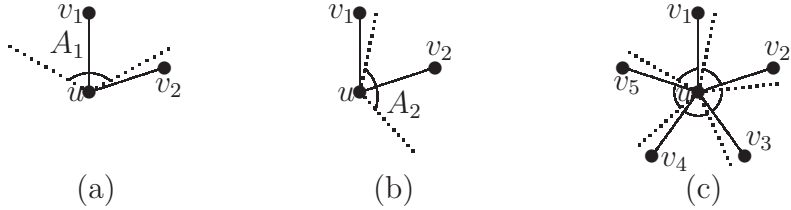
\begin{figure}[!hbtp]
\begin{center}
\begin{picture}(70,60)
\put(30,60){\circle*{5}}
\put(58.5,39.3){\circle*{5}}
\put(30,30){\circle*{5}}
\qbezier(30,60)(30,45)(30,30)
\qbezier(58.5,39.3)(44.3,34.7)(30,30)
{\linethickness{0.3mm}\qbezier[11](0,47.3)(15.0,38.7)(30,30)\qbezier[11](60,46.3)(45.0,38.2)(30,30)}
\qbezier(22.2,34.5)(30,40)(37.8,34.5)
\put(21.5,26.5){$u$}\put(18,58){$v_1$}\put(55,30){$v_2$}\put(15,45){$A_1$}
\put(25,-10){(a)}
\end{picture}
\hskip 1cm \begin{picture}(70,70)
\put(30,60){\circle*{5}}
\put(58.5,39.3){\circle*{5}}
\put(30,30){\circle*{5}}
\qbezier(30,60)(30,45)(30,30)
\qbezier(58.5,39.3)(44.3,34.7)(30,30)
{\linethickness{0.3mm}\qbezier[11](36.2,59.3)(33.1,44.7)(30,30)\qbezier[11](50.1,7.7)(40,18.8)(30,30) }
\qbezier(31.9,38.8)(38.6,32.8)(36.0,23.3)
\put(21.5,26.5){$u$}\put(18,58){$v_1$}\put(55,44){$v_2$}\put(40,22){$A_2$}
\put(35,-10){(b)}
\end{picture}
\hskip 1.9cm \begin{picture}(70,70)
\put(30,60){\circle*{5}}
\put(1.5,39.3){\circle*{5}}
\put(12.4,5.7){\circle*{5}}
\put(47.6,5.7){\circle*{5}}
\put(58.5,39.3){\circle*{5}}
\put(30,30){\circle*{5}}
\qbezier(30,60)(30,45)(30,30)
\qbezier(1.5,39.3)(15.8,34.7)(30,30)
\qbezier(12.4,5.7)(21.2,17.9)(30,30)
\qbezier(47.6,5.7)(38.8,17.9)(30,30)
\qbezier(58.5,39.3)(44.3,34.7)(30,30)
{\linethickness{0.3mm}
\qbezier[11](36.2,59.3)(33.1,44.7)(30,30)
\qbezier[11](4.0,45)(17.0,37.5)(30,30)
\qbezier[11](7.7,9.9)(18.9,20)(30,30)
\qbezier[11](42.2,2.6)(36.1,16.3)(30,30)
\qbezier[11](59.8,33.1)(46,31.6)(30,30)}
\qbezier(31.9,38.8)(37,37)(38.6,32.8)
\qbezier(22.2,34.5)(26,38)(30,39)
\qbezier(23.3,24)(20,28)(21.4,32.8)
\qbezier(33.7,21.8)(29,20)(24.7,22.7)
\qbezier(39.0,30.9)(39,26)(35.3,22.7)
\put(21.5,26.5){$u$}\put(18,58){$v_1$}\put(55,44){$v_2$}\put(48,10){$v_3$}\put(16,2){$v_4$}\put(-2,31){$v_5$}
\put(25,-10){(c)}
\end{picture}
\end{center}
\caption{An illustration of counting $|N_F(u)|$. In $(a)$, the set of nodes falling into the area between the two dashed lines (except node $v_1$) is $A_1$ and $|A_1|\leq k-1$. $v_2$ is the first node with $\angle v_1uv_2\geq \pi/3$. In $(b)$, the set of nodes falling into the area between the two dashed lines, except node $v_2$, is $A_2$ and $|A_2|\leq k-1$. In $(c)$, each angle symbol between a dashed line and a solid line indicates an angle of $\pi/3$. The remaining area contains five narrow angles, each of which is $\pi/15$. All nodes of $N_F(u)$ lie in the narrow angles. This figure shows that upper bound $5k$ can be reached.}\label{fig15-8-5-1}
\end{figure}

It should be noticed that $A_i$ might have an overlap with $A_{i+1}$. One question is whether we can make use of such an overlap to decrease the upper bound. In Fig.\ref{fig15-8-5-1}$(c)$, suppose $N_F(u)$ lie in those narrow angles each of which is bounded by a solid line and a dashed line, and each narrow angle contains exactly $k$ neighbors of $u$ in $F$. Such a configuration does not violate the claim in the proof of Lemma \ref{lemma1}. So, $5k$ cannot be improved based on the claim. Whether there exists some other method to improve the upper bound of $|N_F(u)|$ remains to be further explored.

For $k=2$, the above upper bound $5k$ can be improved to 5. For this purpose, we first prove the following lemma.

\begin{lemma}\label{lemma2}
Let $G$ be a $2$--connected unit disk graph and $F$ be a 2-MSS of $G$.

$(\romannumeral1)$ For any node $u\in F$ with $d_F(u)\geq 3$, no two of its neighbors are adjacent in $F$.

$(\romannumeral2)$  The angle between any two adjacent edges, that meet at a node with degree at least 3 in $F$, is at least $\pi/3$.

$(\romannumeral3)$  If a node $u\in F$ with $d_F(u)\geq 3$ has two neighbors $v$ and $u'$ with $\angle vuu'=\pi/3$, then $\|uv\| = \|uu'\|$. In this case, we can delete either $uv$ or $uu'$ and add edge $vu'$, obtaining another 2-MSS of $G$.
\end{lemma}

\begin{proof}
We use notations and terminologies in the proof of Lemma \ref{lemma1}.

\vskip 0.2cm {\bf Claim.} For any node $u$ with $d_F(u)\geq 3$, $N_F(u)\cap S_u=\emptyset$.

Suppose the claim is not true. Let $u'$ be a node in $N_F(u)\cap S_u$. Since $B_{u}$ is $2$--connected (by property $(c)$ of Lemma \ref{lemma1}), there is an $(u,u')$-path $P$ in $B_{u}-uu'$. Since $B_u$ is a leaf block, there is an $(u',v)$-path $Q$ in $F-(B_{u}\setminus S_u)$. Because $|S_u|=k-1=1$, so $u'$ is the only node in $S_u$, and thus path $Q$ and path $P$ are internally disjoint. Concatenating $Q$ with edge $uv$, we have a $(u,u')$-path in $F-uu'$ which is internally disjoint with path $P$ in $F$. Hence $F-uu'$ has two internally-disjoint $(u,u')$-paths. On the other hand, applying property $(a)$ to edge $uu'$, we see that there is only one internally disjoint path between $u$ and $u'$ in $F-uu'$, a contradiction.

Now, we prove the three properties of this lemma.

$(\romannumeral1)$  Suppose $u\in F$ has $d_F(u)\geq 3$, and $v,u'$ are two neighbors of $u$ which are adjacent in $F$. By Property $(b)$ of Lemma \ref{lemma1}, node $u'$ is in $B_{u}\cap B_{v}$. But then $u'$ is in a 1-separator of $F-uv$, contradicting the above claim.

$(\romannumeral2)$ Suppose two edges $uv$, $uu'$ meet at node $u$ with an angle of less than $\pi/3$. Assume, without loss of generality, that $\|uv\|\geq \|uu'\|$. In this case, $\|uv\| > \|vu'\|$. Since $u'\in B_u\setminus S_u$ and $v\in B_v\setminus S_v$, adding edge $vu'$ merges all blocks of $F-uv$ into one $2$--connected graph. So $F^{'}=F-uv+vu'$ is a $2$--connected spanning subgraph of $G$ with shorter total edge length, contradicting the minimality of $F$.

$(\romannumeral3)$  Suppose edges $uv,uu'$ meet at node $u$ with an angle of $\pi/3$. If $\|uv\| \neq \|uu'\|$, without loss of generality, assume $\|uv\| > \|uu'\|$. Then, $\|uv\|>\|vu'\|$. Similarly to the above, $F'=F-uv+vu'$ is a $2$--connected spanning subgraph of $G$ with shorter total edge length than $F$, a contradiction. Thus $\|uv\|=\|uu'\|$. It follows that $uvu'$ is an equilateral triangle. Then, similar argument as the above shows that both $F'=F-uv+vu'$ and $F''=F-uu'+vu'$ are $2$--connected spanning subgraphs of $G$ with the same total edge length as $F$. Property $(\romannumeral3)$ is proved.
\end{proof}

\begin{lemma}\label{lemma3}
Any $2$--connected unit disk graph $G$ has a $2$-connected spanning subgraph $F$ with maximum degree at most five.
\end{lemma}

\begin{proof}
Let $F$ be a 2-MSS of $G$. By property $(\romannumeral2)$ of Lemma \ref{lemma2}, every node in $F$ has degree at most six.

Suppose $u$ is a node of degree six in $F$, whose neighbors are $u_{0},u_{1},...,u_5$, ordered in a clockwise order. Then $\angle u_{i}uu_{i+1}=\pi/3$ for $i=0,1,...,5$, where ``$+$'' is modulo $6$. We claim that
$$
d_{F}(u_{i})\leq 4 \ \mbox{for} \ i=0,1,...,5.
$$
In fact, by property $(\romannumeral3)$ of Lemma \ref{lemma2}, $F'=(F-uu_{i-1}-uu_{i+1})+u_{i}u_{i-1}+u_{i}u_{i+1}$ is a 2-MSS of $G$. So $d_{F}(u_{i})+2 =d_{F'}(u_{i})\leq 6$, and thus $d_{F}(u_{i})\leq 4$.

Also by property $(\romannumeral3)$ of Lemma \ref{lemma2},  $F'=F-uu_{i}+u_{i}u_{i-1}$ is a 2-MSS of $G$, in which the degree of $u$ is decreased. Notice that $d_{F'}(v)=d_F(v)$ for any $v\in V(G)\setminus\{u,u_{i-1}\}$ and $d_{F'}(u_{i-1})=d_F(u_{i-1})+1\leq 5$. So such an operation results in a 2-MSS of $G$ in which the number of nodes of degree six is strictly decreased. Repeatedly executing such an operation eventually results in a minimum length $2$-connected spanning subgraph of $G$ with maximum degree at most five.
\end{proof}

\subsection{Node--Weighted $k$--Connected Steiner Network in Unit Disk Graphs}\label{subsec3.2}

In this subsection, we present an approximation algorithm for the special minimum node-weighted $k$-connected Steiner network problem in which the terminal set form an $m$-fold dominating set with $m\geq k$.

\begin{definition}[Minimum Node--Weighted $k$--Connected Steiner Network (MNW$k$CSN)]\label{def4}
Given a $k$--connected graph $G = (V ,E)$ with non-negative node weight function $c$
and a terminal node set $T\subseteq V$, MNW$k$CSN is to find a node set $C\subseteq V\setminus T$
with the minimum weight $c(C)=\sum_{v\in C}c(v)$
such that $G[T \cup C]$ is $k$--connected.
\end{definition}

Since $T$ is included in any feasible solution of MNW$k$CSN, we may assume that any node in $T$ has weight zero.

Our algorithm makes use of the {\em subset $k$--connected subgraph problem}. In fact, the subset $k$--connected subgraph problem in \cite{Nutov} has a very general form. For the purpose of this paper, we only use a simplified version whose definition is given as follows.

\begin{definition}[Subset $k$--Connected Subgraph (S$k$CS)]\label{def5}
{\rm Let $G=(V,E)$ be a graph with edge weight function $w$, and let $T\subseteq V$ be a terminal set. A {\em subset $k$--connected subgraph of $G$} is a subgraph $F$ of $G$ such that for any pair of nodes $u,v\in T$, there are at least $k$ internally disjoint $(u,v)$-paths in $F$. The {\em subset $k$--connected subgraph problem} is to find a subset $k$--connected subgraph $F$ of $G$ with the minimum edge weight $w(F)=\sum_{e\in E(F)}w(e)$.}
\end{definition}

Our algorithm for MNW$k$CSN is presented in Algorithm \ref{alg1}.

\begin{algorithm}
\caption{Algorithm for MNW$k$CSN}
Input: A $k$--connected graph $G=(V,E)$ with node cost function $c$, a terminal set $T\subseteq V$ which is an $m$-fold dominating set of $G$ with $m\geq k$, and a $\rho$--approximation algorithm $\mathcal A$ for S$k$CS.

Output: A $k$--connected subgraph $F$ of $G$ containing $T$.

\begin{algorithmic}[1]
\State  Construct an instance $(G,T,w)$ of S$k$CS by assigning edge weight function $w$ by defining  $w(uv)=(c(u)+c(v))/2$ for each edge $uv\in E(G)$.
\State Apply algorithm $\mathcal A$ on $(G,T,w)$ to compute a subset $k$--connected subgraph $F$ of $G$.
\State Output $F$.
\end{algorithmic}\label{alg1}
\end{algorithm}

In general, a feasible solution to S$k$CS might perhaps not be a $k$--connected subgraph, because a $k$--connected subgraph requires {\em every pair} of nodes to be connected through $k$ internally disjoint paths (by Menger's Theorem \cite{Bondy}) instead of merely those pairs of nodes in $T$. However, under the assumption that $T$ is an $m$-fold dominating set of $G$ with $m\geq k$, the output $F$ must be $k$--connected. Suppose this is not true. Let $S$ be a $(k-1)$-separator of $F$. Since every pair of nodes in $T$ are connected by at least $k$ internally disjoint paths in $F$, there exists a connected component of $F-S$, say $R$, such that $T\subseteq V(R)\cup S$. Let $u$ be a node in another connected component of $F-S$. Since $u$ has at least $m\geq k$ neighbors in $T$ and $|S|=k-1<m$, node $u$ has a neighbor in $R$, contradicting that $u$ belongs to a connected component of $F-S$ which is different from $R$. This argument shows that the output of Algorithm \ref{alg1} is indeed a solution to MNW$k$CSN.

Next, we analyze the performance ratio of Algorithm \ref{alg1}.

\begin{lemma}\label{lemma5}
Let $F$ be the output of Algorithm \ref{alg1}. Then $c(V(F))\leq w(E(F))$.
\end{lemma}

\begin{proof}
Notice that although the subgraph induced by node set $V(F)$ is $k$-connected, the subgraph $F$ itself, which is an edge induced subgraph, is not so. However, we must have $d_F(u)\geq 2$ for any node $u\in V(F)$. In fact, if $u\in R$, then $d_F(u)\geq k$. If $u\in V(F)\setminus R$, then the reason why $u$ is added is for connection. Since any node with degree 1 in $F$ cannot play such a role, we have $d_F(u)\geq 2$. It follows that
\begin{align*}
w(E(F)) & =\sum_{uv\in E(F)}\frac{c(u)+c(v)}{2}\\
& =\sum_{u\in V(F)}c(u)\cdot\frac{d_F(u)}{2}\\
& \geq \sum_{u\in V(F)}c(u)\\
& = c(V(F)).
\end{align*}
The lemma is proved.
\end{proof}

\begin{theorem}\label{th1}
Under the assumption that $m\geq k$, Algorithm \ref{alg1} computes a solution to MNW$k$CSN on unit disk graph with performance ratio $2.5k\rho$ for $k\geq 3$ and performance ratio $2.5\rho$ for $k =2$, where $\rho$ is the performance ratio for S$k$CS.
\end{theorem}

\begin{proof}
Let $F$ be the output of Algorithm \ref{alg1}, $OPT_{MNWkCSN}$ be an optimal solution to MNW$k$CSN, and $OPT_{SkCS}$ be an optimal solution to the S$k$CS problem constructed in Line 1 of Algorithm \ref{alg1}. Let $\widetilde{F}$ be a minimum length $k$--connected spanning subgraph of $OPT_{MNWkCSN}$. Denote by $\Delta_{\widetilde{F}}$ the maximum degree of $\widetilde{F}$. Clearly, $\widetilde{F}$ is a feasible solution to S$k$CS. Then by Lemma \ref{lemma5},
\begin{align*}
c(V(F)) & \leq w(E(F))\\
& \leq \rho w(E(OPT_{SkCS}))\\
& \leq \rho w(E(\widetilde{F}))\\
& = \rho \sum_{uv\in E(\widetilde{F})}\frac{c(u)+c(v)}{2} \\
& = \frac{\rho}{2}\sum_{u \in V(\widetilde{F})}c(u)d_{\widetilde{F}}(u)\\
& \leq\Delta_{\widetilde{F}}\frac{\rho}{2}\sum_{u\in V(\widetilde{F})}c(u)\\
& =\Delta_{\widetilde{F}}\frac{\rho}{2} \sum_{u\in V(OPT_{MNWkCSN})}c(u) \\
& =\Delta_{\widetilde{F}}\frac{\rho}{2} c(OPT_{MNWkCSN})
\end{align*}
By Lemma \ref{lemma2} and Lemma \ref{lemma3}, $OPT_{MNWkCSN}$ has a minimum length $k$--connected spanning subgraph  $\widetilde{F}$ with maximum degree at most 5k when $k\geq 3$ and at most 5 when $k=2$. The performance ratios follow.
\end{proof}

In \cite{Nutov}, Nutov gave an $O(k^2\ln k)$-approximation algorithm for S$k$CS. For $k=2$, the S$2$CS problem is a special case of the $\{0,1,2\}$-Steiner network problem for which Fleischer gave a $2$-approximation algorithm in \cite{Fleischer}. Hence the MNW$k$CSN problem on unit disk graph admits constant approximation, the performance ratio of which is $O(k^3\log k)$ for $k\geq 3$ and $5$ for $k=2$. 

\subsection{Algorithm for $(k,m)$--MWCDS}\label{subsec3.3}
The algorithm for $(k,m)$--MWCDS is presented in Algorithm \ref{alg2}. Notice that for $m\geq k$, a graph $G$ has a $(k,m)$-CDS if and only if $G$ is $k$-connected. The if part is obvious since the node set of a $k$-connected graph is a trivial $(k,m)$-CDS. To see the only if part, suppose $D$ is a $(k,m)$-CDS of $G$. If $G$ has a separator $S$ with $|S|\leq k-1$, then $D$ must be completely contained in an $S$-component of $G$, and any node outside of this $S$-component cannot have at least $m\geq k$ neighbors in $D$. So, $G$ is $k$-connected. In the following, we assume that the original unit disk graph $G$ is $k$-connected.

\begin{algorithm}
\caption{Algorithm for $(k,m)$--MWCDS, where $m\geq k$.}
Input: A $k$-connected unit disk graph $G=(V,E)$, an $\alpha$--approximation algorithm $\mathcal A$ for minimum weight $m$--fold dominating set and a $\gamma$--approximation algorithm $\mathcal B$ for MNW$k$CSN.

Output: A $(k,m)$--CDS for $G=(V,E)$.

\begin{algorithmic}[1]
\State  Apply algorithm $\mathcal A$ to compute an $m$--fold dominating set $D$ of $G$.
\State  Reweigh nodes in $D$ to have weight zero. Then apply algorithm $\mathcal B$ to compute a $k$--connected subgraph $F$ of $G$ on terminal set $D$.
\State  Output $V(F)$.
\end{algorithmic}\label{alg2}
\end{algorithm}

To analyze the performance ratio of Algorithm \ref{alg2}, we need the following lemma which is well known in graph theory.

\begin{lemma}[\cite{Bondy}]\label{lem14-6-9-1}
Suppose $G_1$ is a $k$-connected graph and
$G_2$ is obtained from $G_1$ by adding a new
node $u$ and joining $u$ to at least $k$ nodes
of $G_1$. Then $G_2$ is also $k$-connected.
\end{lemma}

\begin{theorem}\label{th2}
Algorithm \ref{alg2} has performance ratio $\alpha+\gamma$.
\end{theorem}

\begin{proof}
Let $OPT$ be an optimal solution to $(k,m)$--MWCDS and $F_{min}$ be an optimal solution to MNW$k$CSN on terminal set $D$. Since $m\geq k$ and $OPT$ is $k$--connected, we see from Lemma \ref{lem14-6-9-1} that the induced subgraph $G[OPT\cup D]$ is $k$--connected, and thus $OPT\cup D$ is a feasible solution to MNW$k$CSN on terminal set $D$. Therefore, $c(V(F_{min})\setminus D)\leq c(OPT)$. Then by Theorem \ref{th1},
\begin{align*}
c(V(F)) & =c(D)+c(V(F) \setminus D)   \\
& \leq \alpha c(OPT)+\gamma c(V(F_{min})\setminus D)  \\
& \leq (\alpha+\gamma)c(OPT).
\end{align*}
The performance ratio is proved.
\end{proof}

If $\mathcal A$ is taken to be Algorithm \ref{alg1}, then we have the following performance ratio:

\begin{theorem}\label{th3}
For $m\geq k$, the $(k,m)$--MWCDS problem on unit disk graph has an $(\alpha+2.5k\rho)$--approximation when $k\geq 3$ and an
$(\alpha+2.5\rho)$--approximation when $k=2$, where $\alpha$ is the performance ratio
for minimum weight $m$--fold dominating set problem,
$\rho$ is the performance ratio for subset $k$--connected subgraph problem.
\end{theorem}

Since the best known $\alpha$ is $O(1)$ \cite{Fukunage}. The best known $\rho=O(k^{2}\log k)$  \cite{Nutov}
 when $k\geq 3$ and $\rho=2$ \cite{Fleischer} when $k=2$, we have the following corollary:

\begin{corollary}
For $m\geq k$, the $(k,m)$-MWCDS problem on unit disk graph has a constant approximation algorithm.
\end{corollary}

\section{Conclusion}\label{sec3}

In this paper, we designed a polynomial-time constant-approximation algorithm for the minimum weight $k$--connected $m$-fold dominating set problem ($(k,m)$-MWCDS) in unit disk graphs, where $m\geq k$.
Prior to this work, constant approximation algorithms were known for $k=1$ with weight and for $2\leq k\leq 3$ without weight. However, for $k\geq 4$, whether $(k,m)$-MCDS on unit disk graph admits a constant approximation is a long standing open problem. We answer the problem confirmatively for any fixed integer $k$, even considering weight.

Our algorithm is based on a constant approximation algorithm for the minimum node-weighted Steiner network problem on unit disk graph under the assumption that the terminal set is a $m$-fold dominating set with $m\geq k$. A key result to the performance ratio is a geometric result, saying that every $k$--connected unit disk graph has a $k$--connected spanning subgraph whose maximum degree is upper bounded by a constant. Is it possible to design a better approximation algorithm for the minimum node weighted $k$--connected Steiner network problem directly? If this can be done, then the performance ratio for $(k,m)$-MWCDS can be improved accordingly.

It should be remarked that after this paper was published in \cite{Shi2017}, Zeev Nutov pointed out a flaw in the proof of Lemma \ref{lemma5}, which, in its original version, says that $c(V(F))\leq \frac{2}{k}w(E(F))$, while it can be seen from this updated version that we can only obtain $c(V(F))\leq w(E(F))$. It should also be pointed out that Takuro Fukunaga obtained a constant ratio for $(k,m)$-CDS almost at the same time \cite{Fukunage}. Our algorithm makes use of Nutov's algorithm for S$k$CS, while Fukunaga's algorithm opens such a black-box by using primal-dual method directly. After fixing the above flaw, our ratio is larger than Fukunaga's ratio by a factor of $k/2$.

\section*{Acknowledgment}
This research is supported by NSFC (61222201, 11531011).


\begin{thebibliography}{1}


\bibitem{Ambuhl} C.~Amb\"{u}hl, M.~Erlebach T, M.~Mihal\'{a}k, M.~Nunkesser, Constant-approximation for minimum-weight (connected) dominating sets in unit disk graphs, APPROX, LNCS 4110: 3--14, 2006.

\bibitem{Berman} P.~Berman, G.~Calinescu, C.~Shah, A.~Zelikovsky, Power efficient mMonitoring management in sensor networks, IEEE Wireless Communication and Networking Conf (WCNC'04), Atlanta, 2329--2334, 2004.

\bibitem{Blum} J.~Blum, M.~Ding, A.~Thaeler, X.~Cheng, Connected dominating set in sensor networks and MANETs, Handbook of Combinatorial Optimization, 329--369, 2005.

\bibitem{Byrka} J.~Byrka, F.~Grandoni, T.~Rothvob, L.~Sanita, Steiner tree approximation via
iterative randomized rounding, Journal of the ACM 60(1), 2013. An improved LP-based approximation for Steiner tree, STOC'10, 2010.

\bibitem{Bondy} J.~A.~Bondy, U.~S.~R.~Murty, Graph theory. Springer, New York, 2008.

\bibitem{Cheng} X.~Cheng, X.~Huang, D.~Li, W.~Wu, D.~-Z.~Du, A polynomial-time approximation scheme for minimum connected
dominating set in ad hoc wireless networks, Networks 42:202--208, 2003.

\bibitem{Clark} B.~N.~Clark, C.~J.~Colbourn, D.~S.~Johnson, Unit disk graphs. Annals of Discrete Mathematics, 48: 165--177, 1991.

\bibitem{Dai} F.~Dai, J.~Wu, On constructing  $k$-connected  $k$-dominating set in wireless ad hoc and sensor networks. Journal of Parallel and Distributed Computing, 66(7): 947--958, 2006.

\bibitem{DaiYu} D.~Dai, C.~Yu, A $(5+\epsilon)$-approximation algorithm for minimum weighted dominating set in unit disk graph, Theoretical Computer Science 410: 756--765, 2009.

\bibitem{Das} B.~Das, V.~Bharghavan, Routing in ad hoc networks using minimum connected dominating sets. In: ICC¡¯97, Montreal, Canada, 376--380, 1997.

\bibitem{Du} D.~-Z.~Du, K.~-I.~Ko, X.~Hu, Design and analysis of approximation algorithms. Springer, New York,2012.

\bibitem{DuBookCDS} D.~-Z.~Du, P.~J.~Wan, Connected Dominating Set: Theory and Applications.  Springer, New York, 2012.

\bibitem{DuH} H.~Du, L.~Ding, W.~Wu, D.~Kim, P.~M.~Pardalos,J.~Willson, Connected dominating set in wireless networks, in: Handbook of Combinatorial Optimization, Second Edition, ed. by P.~M.~Pardalos, R.~L.~Graham, D.~-Z.~Du, 783--834, 2013.

\bibitem{DuY} Y.~Du, H.~Du, A new bound on maximum independent set and minimum connected dominating set in unit disk graphs. Journal of Combinatorial Optimization doi:10.1007/s10878-013-9690-0.

\bibitem{Erlebach} T.~Erlebach, M.~Mihal\'{a}k, A $(4 + \epsilon)$-approximation for the minimum-weight dominating set problem in unit disk graphs, Approximation and Online Algorithms LNCS 5893: 135--146, 2010.

\bibitem{Fleischer} L.~Fleischer, A $2$--approximation for minimum cost \{0, 1, 2\} vertex connectivity, in Proc. IPCO, 115--129, 2001.

\bibitem{Fukunage} T.~Fukunage, Approximation algorithms for highly connected multi-dominating sets in unit disk graphs, Algorithmica 80: 3270--3292, 2018.

\bibitem{Funke} S.~Funke, A.~Kesselman, U.~Meyer, M.~Segal, A simple improved distributed algorithm for minimum CDS in unit disk graphs, ACM Trans. Sensor Net. 2: 444--453, 2006.

\bibitem{Garg} N.~Garg, J.~K\"{o}nemann, Faster and simpler algorithms for multicommodity flows and other fractional packing problems, FOCS'98 300--309, 1998.

\bibitem{Gao} X.~Gao, Y.~Wang, X.~Li, W.~Wu,) Analysis on theoretical bonds for approximating dominating set problems, Discrete Mathematics, Algorithms and Applications, 1:1 71--84, 2009.

\bibitem{Guha} S.~Guha, S.~Khuller, Approximation algorithms for connected dominating sets. Algorithmica 20: 374--387, 1998.

\bibitem{Huang} Y.~Huang, X.~Gao, Z.~Zhang, W.~Wu, A better constant-factor approximation for weighted dominating set in unit disk graph, Journal of Combinatorial Optimization 18: 174--194, 2009.

\bibitem{Holberg} W.~Holberg, The decomposition of graphs into k-connected components. Discrete Mathematics 109: 133--145,1992.

\bibitem{Kim} D.~Kim, W. Wang, X. Li, Z.~Zhang, W.~Wu, A new constant factor approximation for
computing $3$-connected $m$-dominating sets in homogeneous wireless networks. IEEE INFOCOM'10, 1--9, 2010.

\bibitem{LiJ} J. Li, Y. Jin, A PTAS for the weighted unit disk cover problem. ICALP 2015, 898--909.

\bibitem{LiM} M.~Li, P.~Wan, F.~Yao, Tighter approximation bounds for minimum CDS in wireless ad hoc networks, ISAAC'2009, LNCS, 5878: 699--709, 2009.

\bibitem{Li} Y.~Li, Y.~Wu, C.~Ai, F.~Beyah, On the construction of $k$-connected $m$-dominating sets
in wireless networks. J. Combinatorial Optimization 23: 118--139, 2012.

\bibitem{Nutov} Z.~Nutov, Approximating minimum-cost connectivity problems via uncrossable bifamilies. ACM Transactions on Algorithms, 9(1), 417--426, 2012.

\bibitem{Shang} W.~Shang, F.~Yao, P.~Wan, X.~Hu, On minimum $m$--connected $k$--dominating set problem in unit disc graphs. Journal of Combinatorial Optimization 16: 99--106, 2008.

\bibitem{Shi} Y.~Shi, Y.~Zhang, Z.~Zhang, W.~Wu, A greedy algorithm for the minimum $2$--connected $m$--fold dominating set problem. Journal of Combinatorial Optimization 31: 136--151, 2016.

\bibitem{Shi2017} Y.~ Shi, Z.~Zhang, Y.~Mo, D.-Z. Du, Approximation algorithm for minimum weight fault-tolerant virtual backbone in unit disk graphs, IEEE/ACM Transactions on Networking, 25(2): 925--933, 2017.

\bibitem{Thai} M.~Thai, N.~Zhang, R.~Tiwari, X.~Xu, On approximation algorithms of $k$-connected $m$-dominating sets in disk graphs. Theoretical Computer Science 385: 49--59, 2007.


\bibitem{Wan} P.~Wan, K.~Alzoubi, O.~Frieder, Distributed construction of connected dominating set in wireless ad hoc networks, ACM Springer Mobile Networks and Applications, 9(2): 141--149, 2004. A preliminary version of this paper appeared in IEEE INFOCOM, 2002.

\bibitem{Wan1} P.~Wan, L.~Wang, F.~Yao, Two-phased approximation algorithms for minimum CDS in wireless ad hoc networks, IEEE ICDCS, 337--344, 2008.

\bibitem{Wang} F.~Wang, M.~Thai, D.~-Z.~Du, On the construction of $2$--connected virtual backbone in wireless networks. IEEE Transactions on Wireless Communications 8: 1230--1237, 2009.

\bibitem{Wangw} W.~Wang, D.~Kim, M.~An, W.~Gao, X.~Li, Z.~Zhang, W.~Wu, On construction of quality fault--tolerant virtual backbone in wireless networks. IEEE/ACM Transactions on Networking 21(5): 1499--1510, 2012.

\bibitem{WangWei2} W.~Wang, B.~Liu, D.~Kim, D.~Li, J.~Wang, Y.~Jiang, A better constant approximation for minimum 3-connected $m$-dominating set problem in unit disk graph using Tutte decomposition, INFOCOM'15, 2015.

\bibitem{Willson} J.~Willson, Z.~Zhang, W.~Wu, D.~-Z.~Du, Fault-tolerant coverage with maximum liftetime in wireless sensor networks, IEEE INFOCOM'15, 1361--1372, 2015.

\bibitem{WuW} W.~Wu, H.~Du, X.~Jia, Y.~Li, S.~Huang, Minimum connected dominating sets and maximal independent sets in unit disk graphs, Theor. Comput. Sci. 352(1-3): 1--7, 2006.

\bibitem{Wu} Y.~Wu, F.~Wang, M.~Thai, Y.~Li, Constructing $k$-connected $m$-dominating sets in wireless sensor networks. In: Military communications conference, Orlando, FL, 2007.

\bibitem{Zhang} Z.~Zhang, X.~Gao, W.~Wu, D.~-Z.~Du, A PTAS for minimum connected dominating set in 3-dimensional wireless sensor networks, J. Global Optimization, 45: 451--458, 2009.

\bibitem{ZhangTON} Z.~Zhang, J.~Willson, Z.~Lu, W.~Wu, X.~Zhu, D.~-Z.~Du, Approximating maximum lifetime $k$-coverage through minimizing weighted $k$-cover in homogeneous wireless sensor networks, IEEE/ACM Transactions on Networking, DOI: 10.1109/TNET.2016.2531688.

\bibitem{ZhangInfocom} Z.~Zhang, J.~Zhou, Y.~Mo, D.~-Z.~Du, Performance-guaranteed approximation algorithm for fault-tolerant connected dominating set in wireless networks, INFOCOM'16, 2016.

\bibitem{Zhou} J.~Zhou, Z.~Zhang, W.~Wu, K.~Xing, A greedy algorithm for the fault-tolerant connected dominating set in a general graph. Journal of Combinatorial Optimization, 28(1): 310--319, 2014.

\bibitem{Zou2} F.~Zou, X.~Li, D.~Kim, W.~Wu, Node-weighted Steiner tree approximation in unit disk graphs, Journal of Combination Optimization, 18: 342--349, 2009.

\bibitem{Zou} F.~Zou, Y.~Wang, X.~Xu, H.~Du, X.~Li, P.~Wan, W.~Wu, New approximations for weighted dominating sets and connected dominating sets in unit disk graphs, Theoretical Computer Science 412(3): 198--208, 2011.


\end{thebibliography}
\end{document}